 \documentclass[preprint]{elsarticle}

\usepackage{amssymb}
\usepackage{amsthm}
\usepackage{amsmath}
\usepackage{algorithmic}
\usepackage{algorithm}

\journal{arXiv.org}

\begin{document}


\newtheorem{theorem}{Theorem}
\newtheorem{definition}{Definition}
\newtheorem{obs}[definition]{Observation}
\newtheorem{lemma}[theorem]{Lemma}
\newtheorem{corollary}[theorem]{Corollary}
\newtheorem*{claim}{Claim}

\newcommand{\F}{\mathbb{F}}
\newcommand{\nat}{\mathbb{N}}
\newcommand{\real}{\mathbb{R}}
\newcommand{\field}{\mathbb{F}}
\newcommand{\otrootn}{$\tilde{O}(\sqrt{n})$}
\newcommand{\orootn}{$O(\sqrt{n})$}
\newcommand{\rootn}{$\sqrt{n}$}
\newcommand{\Stack}{\texttt{Stack}}
\newcommand{\NonTerm}{\texttt{NonTerm}}
\newcommand{\CompStack}{\texttt{CompStack}}
\newcommand{\select}{\texttt{SELECT}}
\newcommand{\Compress}{\texttt{FP}}
\newcommand{\FP}{\texttt{FP}}
\newcommand{\EQUALITY}{\textsf{\small{EQUALITY}}}
\newcommand{\DISJOINT}{\textsf{\small{DISJOINT}}}

\newcommand{\stack}{\texttt{stack}}
\newcommand{\nonterm}{\texttt{non\_term}}
\newcommand{\compstack}{\texttt{comp\_stack}}
\newcommand{\nextnonterm}{\texttt{next\_non\_term}}
\newcommand{\nextterm}{\texttt{next\_term}}
\newcommand{\comppart}{\texttt{comp\_part}}

\newcommand{\DEGSEQ}{\mbox{\sf Deg-Seq}}
\newcommand{\EQ}{\mbox{\sc Equal}}
\newcommand{\DISJ}{\mbox{\sc Disj}}
\newcommand{\CFL}{\mbox{\sf \small{CFL}}}
\newcommand{\CPDA}{\mbox{\sf \small{CPDA}}}
\newcommand{\PDA}{\mbox{\sf \small{PDA}}}
\newcommand{\CFG}{\mbox{\sf \small{CFG}}}
\newcommand{\DCFL}{\mbox{\sf \small{DCFL}}}
\newcommand{\DLIN}{\mbox{\sf \small{DLIN}}}
\newcommand{\LLone}{\mbox {\sf \small{LL}$(1)$}}
\newcommand{\LL}{\mbox{\sf \small{LL}}}
\newcommand{\LLk}{\mbox{\sf \small{LL}}$(k)$}
\newcommand{\firstk}{\mbox {\sc \small{First$_k$}}}
\newcommand{\VPL}{\mbox{\sf \small{VPL}}}
\newcommand{\restVPL}{\mbox{rest}{\sf-\small{VPL}}}
\newcommand{\VPA}{\mbox{\sf \small{VPL}}}
\newcommand{\wVPL}{\mbox{\sf wVPL}}
\newcommand{\VPG}{\mbox{\sf VPG}}
\newcommand{\wVPG}{\mbox{\sf wVPG}}
\newcommand{\DLCFG}{\mbox{\sf \small{DL-CFG}}}
\newcommand{\Dyckt}{\mbox{\sf Dyck}$_2$}
\newcommand{\Dyckk}{\mbox{\sf Dyck}$_k$}
\newcommand{\otDyckt}{\mbox{$1$-turn-{\sf Dyck}$_2$}}
\newcommand{\otDyckk}{$1$-turn-\mbox{\sf Dyck}$_k$}
\newcommand{\rank}{rank}
\newcommand{\TCz}{\mbox{\sf TC}$^0$}
\newcommand{\NCo}{\mbox{\sf NC}$^1$}

\begin{frontmatter}



\title{Streaming algorithms for language recognition problems}


\author[yahoo]{Ajesh Babu}
\ead{ajesh@yahoo-inc.com}
\address[yahoo]{Yahoo! Labs}

\author[iitb]{Nutan Limaye}
\ead{nutan@cse.iitb.ac.in}
\address[iitb]{Indian Institute of Technology, Mumbai, India}

\author[tifr]{Jaikumar Radhakrishnan}
\ead{jaikumar@tifr.res.in}

\author[tifr]{Girish Varma\corref{cor1}}
\ead{girish@tcs.tifr.res.in}
\address[tifr]{Tata Institute of Fundamental Research, Mumbai, India}

\begin{abstract} 
We study the complexity of the following problems in the streaming
model.

\paragraph{Membership testing for \DLIN } 
We show that every language in \DLIN\ can be recognised by a
randomized one-pass $O(\log n)$ space algorithm with inverse
polynomial one-sided error,  and by a deterministic $p$-pass
$O(n/p)$ space algorithm.  We show that these
algorithms are optimal.

\paragraph{Membership testing for \LL$(k)$} 
For languages generated by \LL$(k)$ grammars with a bound of $r$ on
the number of nonterminals at any stage in the left-most derivation,
we show that membership can be tested by a randomized one-pass
$O(r\log n)$ space algorithm with inverse polynomial (in $n$)
one-sided error.

\paragraph{Membership testing for \DCFL } 
We show that randomized algorithms as efficient as the ones described
above for \DLIN\ and $\LL(k)$ (which are subclasses of \DCFL) cannot
exist for all of \DCFL : there is a language in \VPL\ (a subclass of
\DCFL ) for which any randomized $p$-pass algorithm with error bounded
by $\epsilon < 1/2$ must use $\Omega(n/p)$ space.

\paragraph{Degree sequence problem} 
We study the problem of determining, given a sequence $d_1, d_2,
\ldots, d_n$ and a graph $G$, whether the degree sequence of $G$ is
precisely $d_1, d_2, \ldots, d_n$.  We give a randomized one-pass
$O(\log n)$ space algorithm with inverse polynomial one-sided error
probability. We show that our algorithms are optimal.
\\ \\
Our randomized algorithms are based on the recent work of Magniez et
al. ~\cite{MMN09}; our lower bounds are obtained by considering
related communication complexity problems.
\end{abstract}


\end{frontmatter}


\section{Introduction}
\label{sec:intro}
Modeling computational problems as language recognition is\linebreak
well-established in theoretical computer science.  By studying the
complexity of recognising languages, one seeks to understand the power
and limitations of various computational models, and also classify
problems according to their hardness.  In this paper, we study
language recognition problems in the data stream model.

The data stream model was invented to understand issues that arise in
computations involving large amounts of data, when the processors have
limited memory and are allowed limited access to the input (typically,
restricted to a small number of passes over it). Such a situation
arises when the input is in secondary storage and it is infeasible to
load it all in the main memory. In recent years, this model has gained
popularity for modeling the actions of routers and other agents on
the internet that need to keep aggregate information about the packets
that they handle; the number of packets is large, and the routers
themselves are allowed only a small amount of memory. In this case,
the final decision needs to be based on just one pass over the input.

In the data stream model, the two main parameters of interest are the
memory available for processing and the number of passes allowed.  An
algorithm is considered efficient if the space it uses is
significantly smaller than the input length (ideally, only
polylogarithmic), and the number passes on the input is small
(ideally, just one).  Given these constraints, most interesting
problems become intractable in this model if the algorithm is required
to be {\em deterministic}.  Randomness, however, is remarkably
effective, and many interesting randomized algorithms have been
proposed (starting with Alon et al. \cite{Alon96} and see the survey
by Muthukrishnan~\cite{muthu03}).

When the number of passes over the input is not restricted, or when
random access to the input is available, the data stream model
corresponds closely to the model of space bounded Turing machines.
Often, techniques developed for such unrestricted space bounded
computations, carry over to the data stream model with limited access
to inputs (e.g.  Nisan's pseudorandom generator~\cite{Nisan90}
designed for derandomizing space bounded randomized computations, has
been effectively employed in many data stream algorithms, starting with
Indyk~\cite{Indyk06}).  In this paper, we consider streaming
algorithms for several language recognition problems that can be
solved in $polylog(n)$ space on a Turing machine.


We will assume that the reader is familiar with basic formal language
theory, in particular, the class of {\em context free languages}
(\CFL).  Our results concern some subclasses of {\CFL}s, namely
$\DLIN$, $\LL(k)$ and $\DCFL$ (we recall their definitions in
Sections~\ref{sec:dlin}, \ref{sec:llk} and \ref{sec:lowerbounds}).
Slightly differing definitions for $\DLIN$ were first given by Ibarra
et al.~\cite{IJR88} and Nasu et al.~\cite{NH69}; the definition we use is
due to Higuera et al.~\cite{Hig02}, where the several similar
definitions are compared and a more general class is defined. It was shown by Holzer et al.~\cite{HL93}
that membership in these languages can be tested in space $O(\log n)$.
{\LLk} languages were defined by Lewis et al. and Knuth \cite{LS68,Knuth71}, and they play
an important role in parsing theory.  Informally, they are the
languages for which the left-most derivation can be obtained
deterministically by making a single pass on the input from left to
right with $k$-lookaheads.  Apart from some technicalities arising
from $\epsilon$-rules in the grammar, the class {$\LL(k)$} includes
$\DLIN$.  It was shown by \cite{Coo79} that all deterministic context-free
languages can
be recognised in space $O(\log^2 n)$.  In this paper, we examine if
languages in $\DLIN$ and $\LL(k)$ admit similar efficient membership
testing in the streaming model.

Our work is motivated by a recent membership testing algorithm of
Magniez et al. \cite{MMN09} for the language \Dyckt, which is the
language of balanced parentheses on two types of parentheses.  The
algorithm uses $O(\sqrt{n} \log n)$ space.  We apply their fingerprinting 
based method to the subclass \DLIN\ and also give a deterministic $p$-pass, $O(n/p)$ space algorithm.
\begin{theorem}
\label{thm:dlin-upperbound}
For every $L \in \DLIN$, 
\begin{enumerate}
 \item there is a randomized one-pass $O(\log n)$
space streaming algorithm such that for $x \in \{0,1\}^n$
 \begin{enumerate}
  \item if $x \in L$ then the algorithm accepts with probability $1$;
  \item if $x \notin L$ then the algorithm rejects with probability at least
 $1-\frac{1}{n}$. 
 \end{enumerate} 
 \item there is a deterministic one-pass $O(n/p)$ space streaming algorithm for testing membership in $L$.
\end{enumerate}
\end{theorem}
(Note that our result does not generalize the result of \cite{MMN09}
for {\Dyckt}, because {\Dyckt\ does not belong to \DLIN}.) However,
Theorem \ref{thm:dlin-upperbound} cannot be improved.
\begin{theorem}
\label{thm:dlin-lowerbound}
Let \[ \otDyckt=\{ w \bar{w}^R: w \in \{ (, [ \}^n,\  n\geq 1 \} ,\]
where $\bar{w}$ is the string obtained from $w$ by replacing
each opening parenthesis by its corresponding closing parenthesis;
$\bar{w}^R$ is the reverse of $\bar{w}$. 
\begin{enumerate}
\item Any $p$-pass randomized streaming algorithm that determines
membership in $\otDyckt$ with probability of error 
bounded by $\epsilon < \frac{1}{2}$, must use
$\Omega((\log n)/p)$ space.
\item Any $p$-pass deterministic streaming algorithm that determines
membership in $\otDyckt$ must use $\Omega(n/p)$ space.

\end{enumerate}

\end{theorem}
This result is obtained by deriving, from the streaming algorithm, a
two-party communication protocol for determining if two strings are
equal, and then appealing to known lower bounds for the communication
problem.

We next investigate if efficient membership testing is possible for
languages in classes larger than {\DLIN}. Similar, fingerprinting
based algorithms apply to the class $\LL(k)$, but their efficiency depends
on a certain parameter based the underlying grammar $G$. In order to
state our result precisely, we now define this parameter.

Let $L$ be a language generated by an $\LL(k)$ grammar $G$.  For a
string $w \in L$, let $\rank_G(w)$ denote the maximum number of
nonterminals in any sentential form arising in the (unique) leftmost
derivation generating $w$.  Let the rank of the grammar, $rank_G: \nat
\rightarrow \nat$, be defined as $\rank_G(n)= \max_{w \in \{0,1\}^n
\cap L(G)} \rank_G(w)$.  We will assume that $\rank_G(n)$ is a
well-behaved function, say it is log-space computable.

\begin{theorem}
\label{thm:llk}
Let $G$ be an $\LL(k)$ grammar. 
There is a randomized 
one-pass streaming algorithm that given an input $w \in \{0,1\}^n$
and a positive integer $b$, using space $O(b \log n)$(the dependence on $k$ varies with the grammar),
\begin{enumerate}
\item accepts with probability $1$ if $w \in L(G)$ and $\rank_G(w)
\leq b$; 
\item rejects with probability at least $1-\frac{1}{n}$ if $w
\notin L(G)$ or $\rank_G(w) > b$.
\end{enumerate}
\end{theorem}

\begin{corollary}
Let $L$ be a language generated by an $\LL(k)$ grammar $G$.
There is a randomized 
one-pass streaming algorithm that given an input $w \in \{0,1\}^n$,
using space $O(\rank_G(n) \log n)$,
\begin{enumerate}
\item accepts with probability $1$ if $x \in L$; 
\item rejects with probability at least $1-\frac{1}{n}$ if $x\not\in L$.
\end{enumerate}
\end{corollary}

Note that the above result does not give efficient streaming
algorithms unconditionally, for the space required depends on
$\rank_G(x)$, which in general may grow as $\Omega(n)$.  Note,
however, that results based on such properties of the derivation have
been considered in the literature before. In fact, the class of left derivation bounded languages defined 
by Walljasper~\cite{Wall74}, consists precisely of languages for which $\rank_G(n)$ is a constant independent 
of $n$; this class was also shown to be closed under AFL operations in~\cite{Wall74}. Many well-studied classes of languages are subclasses
of left derivation bounded languages.  The nonterminal bounded
languages generated by nonterminal bounded grammars (which have a
bounded number of nonterminals in the sentential forms in \emph{any}
derivation) were studied by Workman~\cite{Work76}, who also proved
that they contain all ultralinear languages, which are languages
accepted by finite turn pushdown automata (defined by Ginsburg et al.~\cite{Gins66}).
However, nonterminal bounded grammars need not be \LLk.

Despite the dependence on $\rank_G(n)$, the above corollary is
 applicable to classes of languages such as \restVPL\ defined in~\cite{BLV10} (this considered the restriction of {\VPL}s which have
      {\LLone} grammars) and {\DLIN}s restricted to grammars without
      derivations of the form $A\rightarrow \epsilon$. For these
classes $\rank_G(n)$ is bounded by a constant independent of $n$.

%
%
%
%
%
%
%
%
%

We now turn to show that classes provably do not admit solutions
in the streaming model with polylogarithmic space.  Lower bounds for
membership testing of context-free languages in the streaming model
were studied by Magniez et al.~\cite{MMN09}.  They proved that any one-pass
randomized algorithm requires $\Omega(\sqrt{n\log n})$ space for
testing membership in \Dyckt.  More recently, Jain et al.~\cite{JainN10} proved
that if the passes on the input are made only from left to right, then
in spite of making $p$ passes on the input,
the membership testing for \Dyckt\ requires $\Omega(\sqrt{n}/p)$ space.
Here we prove that in general for languages in \DCFL, no savings in
space over the trivial algorithm of simulating the \PDA\ can be expected.
\begin{theorem}
\label{thm:vpl-lowerbound}
There exists a language $L \in \VPL \subseteq \DCFL$ such that any randomized $p$-pass
streaming algorithm requires $\Omega(n/p)$ space for testing membership in
$L$ with probability of error at most $\epsilon < \frac{1}{2}$.
\end{theorem}
The language $L$ in the above result is a slight modification of
\Dyckt. This result is proved by reducing the membership problem in
the streaming model to the two-party communication problem of checking
whether two subsets of an an $n$-element universe are disjoint. 
%

The upper bounds above show that the method of fingerprinting can
be fruitfully applied to many problems to check equality of elements
located far away in the input string. We provide one more illustration
of the amazing power of this technique.

\paragraph{Degree-Sequence, \DEGSEQ:} The degree sequence problem
is the following.
\begin{description}
\item[Input:] A positive integer $n$ and sequence of directed edges
$$(u_1,v_1),(u_2,v_2),\ldots,(u_m,v_m) \text{ where } u_i,v_i \in \{1, 2, \ldots, n\}$$
on vertex set $\{1, 2, \ldots, n\}$.
\item[Task:] Determine if vertices $1, 2,\ldots, n$ 
             have out-degrees $d_1, d_2, \ldots, d_n$, respectively?
\end{description}

This problem is known to be in log-space (in fact in \TCz\ (see for example
\cite{voltext})). It has been observed \cite{DSLN09,FKSV02} that the complexity of graph
problems changes drastically depending on the order in which the input
is presented to the streaming algorithm.  If the input to \DEGSEQ\ is
such that the degree of a vertex along with all the edges out of that
vertex are listed one after the other, then checking whether the graph
has the given degree sequence is trivial.  If the degrees sequence is
listed first, followed by the adjacency list of the graph then we
observe that a one-pass deterministic algorithm needs $\Omega(n)$
space to compute \DEGSEQ.  For a more general ordering of the input
where the degree sequence is followed by a list of edges in an
arbitrary order, we prove the following theorem:

\begin{theorem}
  \label{thm:deg-seq}
  If the input is a degree sequence followed by a list of edges in an arbitrary
  order, then \DEGSEQ\ can be solved 
  \begin{enumerate} 
    \item by a one-pass, $O(\log n)$ space randomized streaming algorithm such
    that if vertices $1,2,\ldots, n$ have out-degrees $d_1, d_2, \ldots,
    d_n$, respectively, then the algorithm accepts with probability $1$ and
    rejects with probability $1-\frac{1}{n}$, otherwise. 
    \item by a $p$-passes, ${O}(( n \log n) /p)$-space deterministic streaming 
    algorithm. 
  \end{enumerate}
\end{theorem}

We also show that the above result is optimal up to a $\log n$ factor.

\begin{theorem}\
  \label{thm:deg-seq-lowerbound}
  \begin{enumerate} 
\item Any $p$-pass randomized streaming algorithm for  \DEGSEQ\ with probability of error bounded by $\epsilon < \frac{1}{2}$, must use
$\Omega((\log n)/p)$ space.
\item Any $p$-pass deterministic streaming algorithm for \DEGSEQ\ must use $\Omega(n/p)$ space.
  \end{enumerate}
\end{theorem}


\section{Membership testing of \DLIN} 
\label{sec:dlin}
In this section, we study the complexity of membership testing for a
subclass of context free languages called \DLIN, in the streaming model.  Informally it is the class
of languages accepted by $1$-turn \PDA(i.e. \PDA\ which do not make a push move
after having made a pop move), with restrictions similar to \LL$(1)$.

We start with some definitions.  See \cite{HMU06} for the basic definitions
regarding context-free grammars(\CFG) and pushdown automata(\PDA).

\begin{definition}[Higuera et al.\cite{Hig02}] 
\textbf{Deterministic linear CFG} or \DLCFG, is a \CFG\ $(\Sigma, N, P,S)$ for
which, every production is of the form $A \rightarrow a \omega$ or $A
\rightarrow \epsilon$, where $a \in \Sigma$ and $\omega \in (N \cup \{ \epsilon
\})\Sigma^*$ and for any two productions, $A \rightarrow a \omega$ and $
B\rightarrow b \omega'$, if $A=B$ then $a\neq b$, where $a,b \in \Sigma$ and
$\omega, \omega' \in (N \cup \{ \epsilon \})\Sigma^*$.
\end{definition}

\begin{definition} 
\textbf{Deterministic linear CFL}, \DLIN, is the class of languages for which
there exists a \DLCFG\ generating it.
\end{definition}

\DLIN\ is a well studied class in language theory.  Higuera et al.
\cite{Hig02} gives algorithms for learning such grammars.  Many variations of the
above definition have been considered in earlier works.  The above definition is
more general than the ones given in \cite{IJR88,HL93} as was proved in
\cite{Hig02}.  Note that the set of languages accepted by deterministic 1-turn
\PDA\ is a strict super-set of \DLIN.  For example $L = \{a^nb^n \text{ or }
a^nc^n \mid n >0\} \notin \DLIN$ but is accepted by a deterministic 1-turn \PDA
.


\begin{definition} 
\textbf{Canonical Pushdown Automaton} or \CPDA\ for a language $L$ generated by
a \CFG\ $G=(\Sigma, N, P,S)$ is a \PDA\ $M_L=(Q = \{q\},\Sigma,\Gamma = N \cup
\Sigma, \delta,q_0 = q,S)$ , where the transition function $\delta$ is defined
as follows:
\begin{enumerate}
  \item for each production of the form $A\rightarrow a\omega$ where $\omega \in (N\cup \Sigma)^*$,
  $\delta(q,a,A)=(q,\omega)$\footnote{$\delta(q,a,A)=(q,\omega)$ implies that
  when the \PDA\ is at state $q$, has $a$ as the next input symbol and $A$ on
  top of stack, will remain in state $q$, replacing $A$ by $\omega$ at the top
  of the stack.}.
  \item for every production $A \rightarrow \omega$ that is not considered above, 
  $\delta(q,\epsilon, A) = (q,\omega)$.
  \item for all $a\in \Sigma$, $\delta(q,a,a)=(q,\epsilon)$.
\end{enumerate}
$M_L$ starts with only the start symbol $S$ on the stack and it accepts by empty
stack. The language accepted by $M_L$ is $L$.
\end{definition}

If the rules of the form $A \rightarrow \epsilon$ are removed from a \DLCFG,
then the corresponding \CPDA\ is deterministic.  However if the length of the
string is known before hand, then we can infer when such a rule is to be
applied.  It is precisely when the sum of the length of the string seen so far
and the number of nonterminals in the stack add up to the total length.  So the
\CPDA\ can be simulated deterministically by making only a single pass over
the input, but the stack can take up $\Omega(n)$ space.  The algorithm for
membership testing of \DLIN\ is obtained by simulating the \CPDA\ with a
compressed stack.  The stack is compressed by using a hash function which is a
random evaluation of a polynomial constructed from the stack.  This method
commonly known as fingerprinting (see \cite{MR96}, Chapter 7) was used by Magniez et al.
\cite{MMN09}, for giving a streaming algorithm for membership testing of \Dyckt.
Here we apply the technique to the class of \DLIN .  Note that \Dyckt\ is not
contained in \DLIN.

\subsection{Compressing the stack}
\label{dlin-comp}

First we make an observation about the stack of a \CPDA\ for a language $L$,
generated by a \DLCFG.

\begin{obs}
\label{obs:stack} 
For any string $w \in L$ and at any step $i \in [|w|]$, the stack of the \CPDA\
contains at most one nonterminal.
\end{obs}

Consider the run of the \CPDA\ on $w \in \Sigma^*$ in which any transition of
the form $\delta(q, \epsilon, A) = (q,\epsilon)$ is applied only at the step $i$
when the sum of $i-1$ and the number of terminals in the stack adds up to $|w|$.
For $w \in \Sigma^*$, $i \in [|w|]$, let $\Stack(w,i) \in \Sigma^*$ be the
sequence of terminals in the stack of the \CPDA, when it encounters the $i$th
symbol of input $w$.  We consider it from bottom to the first nonterminal or the top
if there is no nonterminal.  If the \CPDA\ rejects before reaching $i$, then
$\Stack(w,i)$ is not defined.  Similarly, let $\NonTerm(w,i)$ be the unique
nonterminal on top of the stack of the \CPDA, when it has reached position $i$
on input $w$.  If there is no nonterminal in the stack, then it is $\epsilon$.

We will assume a fixed bijective map from $\Sigma = \{a_1,a_2, \ldots, a_m\}$ to $[m]
= \{1,2, \ldots, m\}$.  Furthermore we will use ${\bf a_i}$ to denote the value of
the map on $a_i$.  For any string $v\in \Sigma^*$, a prime $p$, formal variable
$x$, let $$\Compress(v,h,x,p)= \sum_{j=1}^{|v|} {\bf v[j]} x^{h+j-1} \mod p$$ be
a polynomial over $\field_p$.  Then $$\CompStack(w,i,x) =
\Compress(\Stack(w,i),0,x,p)$$ can be considered as an encoding of
$\Stack(w,i)$.

\begin{obs}
$\CompStack(w,n,x)$ has degree at most $n$ and is the zero polynomial if and
only if $w \in L$.
\end{obs}

It is therefore sufficient to check whether $\CompStack(w,n,x)$ is the
zero polynomial for testing membership in $L$.  To explicitly store this
polynomial, $\Omega(n)$ space may be required.  But a random evaluation of a
non-zero degree $d$ polynomial over $\mathbb{F}_p$ is zero with probability
at most $d/p$ (due to Schwartz Zippel Lemma).  Hence it suffices to keep a random
evaluation of $\CompStack(w,i,x)$, which can be stored using just $\lceil \log p
\rceil$ bits, for checking if it is zero.  If $p = O(n)$ then the space needed
is considerably reduced to $O(\log n)$.

\subsection{Algorithm}
The algorithm (Algorithm \ref{alg1}) is obtained by observing that the \CPDA\ can be simulated using a
compressed stack.
\begin{algorithm}
\caption{Randomized one pass algorithm} 
\label{alg1}
\begin{algorithmic}[1]
\STATE Input : $w \in \Sigma^*$. Let $|w|=n$. 
\STATE Pick $\alpha$ uniformly at random from $\mathbb{F}_p$. 
\STATE $\compstack \leftarrow 0; ~ \nonterm \leftarrow S; ~ h \leftarrow 0 $ 
\FOR{$i = 1$ to $n$} 
  \IF {$\nonterm \neq \epsilon$} 
    \IF{$h+i-1 = n$} 
    \STATE \textbf{if} a rule of the form $\nonterm \rightarrow \epsilon$ does not exist \textbf{then} reject
    \STATE {\bf else} $\nonterm \leftarrow \epsilon$
    \ELSE 
    \STATE Find the unique rule of the form below. Otherwise reject
     $$\nonterm \rightarrow w[i]Bv,~ v \in \Sigma^*,
    B \in N \cup \{ \epsilon \} $$ 
    \STATE $\compstack \leftarrow \compstack + \Compress(v^R,h,\alpha,p) \mod p$
\\ \COMMENT{where $v^R$ is $v$ reversed}
    \STATE $\nonterm \leftarrow B~; ~h \leftarrow h+ |v|$ 
    \ENDIF
  \ELSE
    \STATE $\compstack \leftarrow \compstack - {\bf w[i]}\alpha^{h-1} \mod p$
    \STATE $h \leftarrow h - 1$
  \ENDIF
\ENDFOR
\STATE \textbf{if} $\compstack=0 \text{ and } h = 0$ \textbf{then} accept
\STATE \textbf{else} reject
\end{algorithmic}
\end{algorithm}

Algorithm \ref{alg1} uses $\lceil \log p \rceil$ bits to store $\alpha$,
$\lceil \log p \rceil$ for $\compstack$, $2 \lceil \log n \rceil$ for $i,h$ and
some constant space that depends on the grammar for $\nonterm$.  Hence the space complexity is $2\lceil \log p \rceil + 2\lceil \log n \rceil + c$.  
 It also uses $\lceil \log p \rceil$ random bits. 
 
\subsection{Proof of Correctness}

\begin{lemma}
\label{lem:comp-sim}
If the input is not rejected on or before the $i^{\text{th}}$ iteration of for
loop on line 4 of algorithm \ref{alg1} then
\begin{itemize}
\item $h=|\Stack(w,i)|$
\item $\compstack = \CompStack(w,i,\alpha)$ 
\item $\nonterm = \NonTerm(w,i)$.
\end{itemize}
\end{lemma}
\begin{proof}
The lemma is proved using induction on $i$.  At $i=1$, $h=0,~\compstack =
\CompStack(w,1,\alpha) = 0$ and $\nonterm = \NonTerm(w,1) = S$.  Assuming above
is true for the $i$th iteration of the loop.  After the updates in line 11 (or
15), we have that $\compstack = \CompStack(w,i,\alpha) + \sum_{j=1}^{|v|}
{\bf v^R[j]}\alpha^{h+j-1} \mod p = \CompStack(w,i+1,\alpha)$ (or $\compstack =
\CompStack(w,i,\alpha) - w[i]\alpha^{h-1} \mod p = \CompStack(w,i+1,\alpha)$,
respectively).  Similarly $h$ and $\nonterm$ are updated correctly in lines 8, 12 and 16.
\end{proof}

Applying Lemma \ref{lem:comp-sim} for $i=n$, we get the following corollary.
\begin{corollary}
If $w\in L$ then Algorithm \ref{alg1} accepts with probability 1.
\end{corollary}
\begin{lemma}
If $w\notin L$ then  $\Pr[ \text{Algorithm \ref{alg1} accepts} ] \leq n/p$.
\end{lemma}
\begin{proof}
If $w\notin L$ then \CPDA\ rejects, say at step $j$. There are three cases.
\begin{enumerate}
\item $\NonTerm(w,j)$ was defined and it rejected as a matching rule of the 
 form $\NonTerm(w,j) \rightarrow w[j]\omega,~ \omega \in 
 \Sigma^*(N \cup \{ \epsilon \})\Sigma^*$ could not be found.
 \item $\NonTerm(w,j)$ was not defined and it rejected as the last character of
  $\Stack(w,j)$ was not $w[j]$.
  \item the stack was not empty at the end of the string.
\end{enumerate}

In the case $1$, Algorithm \ref{alg1} rejects with probability 1.  For case $2$,
the monomial subtracted by the algorithm is ${\bf w[j]}\alpha^{h-1}$.  The only
other monomial in the sum with the degree $h-1$ is ${\bf a}\alpha^{h-1}$ where $a$
is the last character of $\Stack(w,j)$.  Also after the $j$th step no monomial
of degree $h$ is subtracted.  So the polynomial for which $\compstack$ is an
evaluation is not the zero polynomial.  This is also true in case $3$, as the
stack is not empty.  The lemma follows, by an application of the Schwartz Zippel
Lemma.
\end{proof}
Theorem \ref{thm:dlin-upperbound} is obtained by finding a prime $p$ between $n^2$ and $2n^2$ by brute force search and then using Algorithm \ref{alg1}.

\subsection{A deterministic multi-pass algorithm}
\label{multi-pass}
In this section we give a deterministic multi-pass algorithm for the membership
testing of any language in \DLIN . This is done by first reducing the membership testing problem for any $L \in \DLIN $ to membership testing of a particular language \Dyckk\ $\in \DLIN$. Recall that \Dyckk\ is the language generated by the grammar 
$$S \rightarrow SS \mid (_1 S )_1 \mid (_2 S )_2 \mid \cdots \mid (_k S )_k \mid \epsilon $$
and \otDyckt\ is generated by
$$S \rightarrow ( S ) \mid [ S ] \mid \epsilon .$$  
We will be using the following definition of streaming reduction:
\begin{definition}[Streaming Reduction]
Fix two alphabets $\Sigma_1$ and $\Sigma_2$.  A problem $P_1$ is
$f(n)$-streaming reducible to a problem $P_2$ in space $s(n)$ if for every
input $x \in \Sigma_1^n$, there exists $y_1 y_2\ldots y_n$ with $$y_i \in
\cup_{i=1}^{f(n)} \Sigma_2^{i} \cup \{\epsilon\}$$ such that:
\begin{itemize}
\item $y_i$ can be computed from $x_i$ using space $s(n)$.
\item From a solution of $P_2$ on input $y$, a solution on $ P_1$ on input $x$
can be computed in space $s(n)$.
\end{itemize}
\end{definition}
Note that our definition is a slight modification of the definition from
\cite{MMN09}\footnote{In \cite{MMN09}, $y_i$ s are assumed to be of fixed
length, i.e. from $\Sigma_2^{f(n)}$}.  In \cite{MMN09}, it was observed that the
membership testing of \Dyckk\ $O(\log k)$-streaming reduces in $O(\log k)$ space
to membership testing of \Dyckt.  We show that the membership testing for any
language in \DLIN\ $O(1)$-streaming reduces in $O(\log n)$ space to membership testing in \otDyckk,
where $k$ is the alphabet size of the language. It is easy to see that in the the reduction of Magniez et al.~\cite{MMN09}, the output of the reduction is in \otDyckt\ if and only if the input is in \otDyckk . Hence we have the following
theorem:
\begin{theorem} 
\label{thm:str-reduce}
The membership testing for any language in \DLIN\ $O(\log |\Sigma|)$-streaming
reduces in $O(\log n)$ space to membership testing in \otDyckt , where $\Sigma$ is 
the alphabet of the language.
\end{theorem}

Say $L$ is a fixed \DLIN, with $\Sigma = \{a_1, a_2, \ldots, a_k\}$.  Given an 
input $w$, the streaming reduction outputs a string $w' \in \Sigma \cup 
\overline{\Sigma}$ so that $w'$ is in \otDyckk\ 
if and only if $w$ belongs to $L$. Here $\overline{\Sigma} = \{\overline{a_1},
\overline{a_2},\ldots,\overline{a_k} \}$ and for each $i \in [k]$ $(a_i, 
\overline{a_i})$ is a matching pair. The streaming reduction is obtained by 
making a change to the steps $11$ and $15$ of Algorithm \ref{alg1} and is given
as Algorithm \ref{alg2}.

\begin{algorithm}
\caption{Streaming reduction from $L \in \DLIN$ to \Dyckk} 
\label{alg2}
\begin{algorithmic}[1]
\STATE Input : $w \in \Sigma^*$. Let $|w|=n$.
\STATE Output : $w' \in \Sigma \cup \overline{\Sigma}$  
\STATE $\nonterm \leftarrow S;~w' \leftarrow \epsilon$ 
\STATE $i \leftarrow 1$
\WHILE{$i\leq n$}
  \IF {$\nonterm \neq \epsilon$}
    \IF{$|w'|+i-1 = n$} 
      \STATE \textbf{if}  a rule of the form $\nonterm \rightarrow \epsilon$ does not exist \textbf{then} reject
      \STATE \textbf{else} $\nonterm \leftarrow \epsilon$
    \ELSE
    \STATE Find the unique rule of the form below. Otherwise reject
     $$\nonterm \rightarrow w[i]Bv,~ v \in \Sigma^*,
    B \in N \cup \{ \epsilon \} $$ 
    \STATE $w'\leftarrow w' \cdot v^R$ 
    \STATE $\nonterm \leftarrow B; ~i \leftarrow i + 1$
    \ENDIF
  \ELSE
    \STATE $w'= w' \cdot \overline{w[i]}$;~ $i \leftarrow i+1$
  \ENDIF
\ENDWHILE
\end{algorithmic}
\end{algorithm}

From Theorem \ref{thm:str-reduce}, we know that any language in \DLIN\ $O(\log
|\Sigma|)$-streaming reduces to \otDyckt.  Thus it suffices to give a
$p$-passes, $O(n/p)$-space deterministic algorithm for membership testing of
\otDyckt.

The algorithm divides the string into blocks of length $n/2p$.  Let the blocks
be called $B_0,$ $B_1,$ $\ldots,$ $B_{2p-1}$ from left to right. (i.e. $B_i =
w[i(n/2p)+1]~w[i(n/2p)+2]$ $\ldots w[(i+1)n/2p]$.) The algorithm considers a
pair of blocks ($B_j$,$B_{2p-(j+1)}$) during the $j$th pass.  Using the stack
explicitly, the algorithm checks whether the string formed by the concatenation
of $B_j$ and $B_{2p-(j+1)}$ is balanced.  If it is balanced, it proceeds to the
next pair of blocks.  The number of passes required is $p$.  Each pass uses
$O(n/p)$ space and the algorithm is deterministic.  Later in Section \ref{sec:lowerbounds} we show that this algorithm is optimal.

\section{Membership Testing of \LLk\ languages}
\label{sec:llk}
In this section we give a randomized streaming algorithm for testing membership
in \LLk\ languages.
Let $G=(N,\Sigma,P,S)$ be a fixed grammar. For a string $w \in \Sigma^*$, let
$$\text{pref}_k(w) = \begin{cases} \text{ if } |w| > k \text{ then the first } k \text{ characters of } w  \\
                            \text{ else  } w 
\end{cases}.$$
The \emph{select set} of a production $A \rightarrow \alpha$, where $A \in N$ and $\alpha \in (N \cup \Sigma)^*$ 
 is $$\select(A 
\rightarrow \alpha) = \{u \mid \exists v,w \in \Sigma^*,  \alpha  v
\text{ derives } w \text{ and } \text{pref}_k(w) = u\}.$$

\begin{definition}[Lewis et al.~\cite{LS68}]
A grammar $G=(N,\Sigma,P,S)$ is \LLk\ if for any two distinct productions of the form $A \rightarrow \alpha,~A \rightarrow \beta$, 
the select sets are disjoint. \LLk\ languages are the class of languages generated by \LLk\ grammars.
\end{definition}

From now on, we describe an algorithm for \LLone\ languages.  It is easy to observe that it
generalises for \LLk\ languages.  Let $L$ be a language generated by an \LLone\
grammar $G$.  It is known that for any two distinct rules $R\neq R'$ in the
production set of $G$ with the same left side, $\select(R)$ and $\select(R')$
are disjoint.  We call this the \LLone\ property.  Note that \DLCFG s with no
epsilon rules have this property.  Therefore, languages generated by \DLCFG s
with no epsilon rules, are a subclass of \LLone.  As noted by Kurki-Suonio~\cite{KS69}, they
are in fact a proper subclass of languages generated by \LLone\ grammars with no
epsilon rules.  As the part of the preprocessing, for every rule $R$ of the
grammar we compute the set $\select(R)$.  This requires only $O(1)$ space as the
grammar is fixed.

Our membership testing algorithm for \DLIN\ uses the \LLone\ property
non trivially.  Algorithm \ref{alg1} can be thought of as working in two main
steps.  The first step involves reading a terminal from the input and deciding
the next rule to be applied.  The second step consists of updating the stack
appropriately.  The \LLone\ property enables the \CPDA\ to deterministically
decide the next rule to be applied having seen the next input terminal.
Therefore, the first step will remain unchanged even in the case of membership
testing of \LLone\ languages.  In what follows we describe the second step.

Let $\Gamma_k \gamma_{k} \ldots \Gamma_0 \gamma_0,~\gamma_i \in
\Sigma^* \text{ and } \Gamma_i \in N$ be any sentential form arising in the
derivation of $w \in L$.  Then the corresponding \CPDA\ will store this in the
stack(in the above order from top to bottom). 
It is easy to see that the \CPDA\ is generating the left most derivation of $w$.  The space efficient algorithm
that we give below compresses the strings $\gamma_i$s as before and stores
$\Gamma_i$, compression of $\gamma_i$ and $|\gamma_i|$ as a tuple on the stack.
For a string $w \in L$, 
the algorithm runs in space $O(rank(w) (\log p + \log n))$, where $p$ is the size of
the field over which the polynomial is evaluated.  

\subsection{Streaming algorithm for testing membership in \LLone\ languages}
Given below is the randomized streaming algorithm for testing
 membership in \LLone\ languages.

\begin{algorithm}
\caption{Randomized one pass algorithm} 
\label{alg3}
\begin{algorithmic}[1]
\STATE Input : $w \in \Sigma^*$. Let $|w|=n$. 
\STATE Pick $\alpha$ uniformly at random from $\mathbb{F}_p$. 
\STATE $\comppart \leftarrow 0 ~; ~ \nonterm \leftarrow S ~; ~ h \leftarrow 0 $ 
\STATE $\compstack.\text{push}(\comppart, \nonterm, h)$
\STATE $i \leftarrow 1$
\WHILE{$i \leq n$ and $\compstack$ not empty} 
  \STATE $(\comppart, \nonterm, h) \leftarrow \compstack.\text{pop}() $ 
  \IF {$\nonterm \neq \epsilon$}
    \STATE Find the unique rule $R$ of the form below such that $w[i] \in \select(R)$. Otherwise reject.
     $$\nonterm \longrightarrow B_t \beta_t B_{t-1} \beta_{t-1}\ldots B_0
\beta_{0} \text{ where all } \beta_i \in \Sigma^*,
    \text{ and } B_i \in N \cup \{ \epsilon \} $$
    \STATE $\compstack.\text{push}(\comppart +\Compress(\beta_0^R,h,\alpha,p),B_0 , h +
|\beta_0|)$
    \FOR{$k\leftarrow 1$ to $t$}
    \STATE $\compstack.\text{push}(\Compress(\beta_k^R,0,\alpha,p),B_k,|\beta_k|)$ \textbackslash\textbackslash ~$\beta_i^R=reverse(\beta_i)$ 
    \ENDFOR
  \ELSE
    \IF{$h\neq 0$} 
      \STATE $ \comppart \leftarrow \comppart -{\bf w[i]}\alpha^{h-1}  \mod p ~;~h \leftarrow h - 1$
      \STATE $\compstack.\text{push}(\comppart, \nonterm, h)$
    \ELSIF{$\comppart \neq 0$}
      \STATE reject
    \ENDIF
    \STATE $i \leftarrow i+1$
  \ENDIF
\ENDWHILE 
\STATE \textbf{if} \compstack\ is not empty \textbf{then} reject \textbf{else} accept
\end{algorithmic}
\end{algorithm}
The algorithm uses $\lceil \log p \rceil$ space to store $\alpha$,
$\lceil \log p \rceil + O(1) + \log n$ space to store a tuple on the
stack. On input $w$, the space used by the 
algorithm is at most $rank_G(w)(\log n + \lceil \log p \rceil + O(1))$. Therefore, for 
a language generated by grammar $G$, the space used by the
algorithm for checking $w \in L$ is at most $O(rank_G(n)(\log n + \log p))$. 
For proving Theorem \ref{thm:llk}, Algorithm \ref{alg3} can be modified to 
to take an additional parameter $b$ as an input and reject when $w \notin L$ or the number of items in the stack exceeds $b$. Also $p$ can be set to a prime between $n^2$ and $2n^2$ which can be found by brute force search, so that error
probability $n/p \leq 1/n$. 

\subsection{Correctness of the algorithm}
\label{subsec:corr-algo2}
In this section we prove the correctness of the algorithm.  Note that, given an
\LLone\ grammar, the simulating \CPDA\ performs a top-down parsing of the
grammar.  On reading a symbol from the input, and the top of the stack, it
deterministically picks a rule to be applied next.  If no such rule exists, it
halts and rejects.  If such a rule is found, it pushes the right hand side into
the stack.  As long as the stack-top is a nonterminal it repeats this process.
If the stack-top is a terminal, it pops the top terminal from the stack,
provided it matches with the next input letter.  Suppose there is a mismatch, it
halts and rejects.  If after processing the whole string the stack is empty, it
accepts.

We now prove that the working of the algorithm has a close correspondence with
the working of the \CPDA.

\begin{lemma}
\label{lem:inductive}
Let $\Stack(t)=\Gamma_k \gamma_{k} \ldots \Gamma_0 \gamma_0,~ \Gamma_i \in N,
\gamma_i \in \Sigma^*$ be the contents of the stack of \CPDA\ before the 
$t^{\text{th}}$ step(counted in terms of application of the transition function) and
$$\stack(t)=[(\comppart_j,\nonterm_j,h_j),\ldots,(\comppart_0,\nonterm_0,h_0)]$$
 be
the contents of the stack of Algorithm \ref{alg3} before the $t^{\text{th}}$
iteration of the while loop in line $6$. 
If the \CPDA\ has not rejected on or before step $t$ then $j=k$ and $\forall i \in
\{0,\cdots k\}$,
\begin{itemize}
\item $\comppart_i = \Compress(\gamma_i,0,\alpha,p)$ \footnote{\Compress\ was 
defined in Section \ref{dlin-comp}}
\item $\nonterm_i = \Gamma_i$
\item $h_i = |\gamma_i|$
\end{itemize}
\end{lemma}
\begin{proof}
The lemma can be proved by induction on $t$.  At $t=1$,
$\Stack(1)=S,~\stack(1)=[(0,S,0)]$(due to the initialisation steps $3,4$) and
the lemma is true.  Suppose it is true at step $t$, we will prove that the lemma
holds at $t+1$st step.  We consider various cases.  Assume that $\Gamma_k \neq
\epsilon$ in $\Stack(t)$.  Therefore, by inductive hypothesis the stack-top
maintained by the algorithm has $\Gamma_k$ in its second component.  Then steps
$9,10,12,16$, makes sure that updates are made correctly.  Suppose
$\Gamma_k=\epsilon$ and $\gamma_k=av,~a\in \Sigma, v\in \Sigma^*$ then by
inductive hypothesis, the top most item of $\stack(t)$ is
$(\Compress(\gamma_k,0,\alpha,p),\epsilon,|\gamma_k|)$.  By definition
$\Compress(\gamma_k,0,\alpha,p) = a\alpha^{|\gamma_k|-1} +
\Compress(v,0,\alpha,p)$.  If $|v| >0$ then after the execution of step $17$,
this will become $(\Compress(v,0,\alpha,p),\epsilon,|v|)$ which is same as the
top most item of $\Stack(t+1)$.  On the other hand if $v=\epsilon$, this item
not push back in to the stack.
\end{proof}

\begin{lemma}
If $w \in L$ then the algorithm accepts with probability $1$. If $w \notin L$
then the probability that the algorithm accepts is bounded by $n/p$.
\end{lemma}
\begin{proof}
If $w\in L$ then by Lemma \ref{lem:inductive} we have that the algorithm always
accepts.  Suppose the \CPDA\ rejects at a certain step $t$, when symbol at the
$t'$th position of the input was accessed. Let the   There are three cases:
\begin{enumerate}
\item \CPDA\ had a non-terminal on the stack-top and it rejected as a 
	matching rule to be applied could not be found.
 \item \CPDA\ had a terminal on the stack-top, say $a$ and it rejected because
$a \neq w[t']$.
  \item the stack was not empty at the end of the string.
\end{enumerate}
In Case $1$, the algorithm rejects with probability $1$.  For Case $2$,
let the top most item in the stack of the algorithm at step $t$ be 
$(\comppart,\nonterm,h)$. Then the algorithm subtracts ${\bf w[j]}\alpha^{h-1}$ from the 
stack and decreases the height by $1$.  The only other monomial in $\comppart$
with degree $h-1$ is ${\bf a}\alpha^{h-1}$. Hence $\comppart$ is a random evaluation of
a nonzero polynomial of degree at most $n$. From Lemma \ref{lem:inductive}, 
$\nonterm = \epsilon$ and hence no other monomial of degree $h$ is added or
subtracted from $\comppart$. Now either the stack item $(\comppart,\nonterm,h)$
is never popped, or at the time of popping $\comppart$ is checked to be zero.
In the former case, the algorithm rejects with probability $1$ and in the latter
with probability at least $1-n/p$. In Case $3$ the algorithm rejects
with probability $1$.
\end{proof}

Now Theorem \ref{thm:llk} follows from the above lemma by appropriately
selecting the value of $p$ to be a prime between $n^{c+1}$ and $2n^{c+1}$.
Such a prime can be obtained in time polynomial in $n$ by exhaustive search. 

\section{Lower bounds for membership testing}
\label{sec:lowerbounds}

In this section, we prove that the algorithms given in Section
\ref{sec:dlin} are optimal.

\begin{proof}[Proof of Theorem \ref{thm:dlin-lowerbound}]

  We reduce the two-party communication problem of testing equality
  ($\forall x,y \in \{0,1\}^n,~\EQUALITY(x,y) = 1 \leftrightarrow x=y$) of
  strings to membership testing for \otDyckt. In this communication problem,
  the first party, Alice, is given a string $x$ and the other party, Bob,
  is given the string $y$, and they need to communicate to determine
  if $x$ and $y$ are equal. 

  Suppose there is a $p$-pass streaming algorithm for \otDyckt\ using
  space $s$.  We will show that such an algorithm leads to protocol
  for the communication problem, where the total communication is
  $(2p-1)s$. First, Alice and Bob transform their inputs as follows.
  Let $x'$ be the string obtained from $x$ by replacing every $0$ by a
  $[$ and every $1$ by a $($; let $y'$ be the string obtained from $y$
  by first reversing it and then replacing $0,1$ by $],)$
  respectively.  Note that the string $z=x'y'\in
  \big\{(,[,],)\big\}^{2n}\in \otDyckt\ $ iff $x=y$. Alice and Bob
  will simulate the streaming algorithm on $z$ in the following
  natural way: Alice runs the streaming algorithm on $x'$ and on
  reaching the end of the her input, passes on the contents of the
  memory to Bob who continues the simulation on $y'$ and passes the
  contents of the memory back to Alice at the end. If there algorithm
  makes $p$ (left to right) passes, then during the simulation the
  contents of the memory change hands $2p-1$ times. If the algorithm
  is deterministic, the protocol is deterministic. If the algorithm is
  randomized, the protocol is randomized and has the same error
  probability.

  Since any deterministic protocol for $\EQUALITY(x,y)$ requires $n$
  bits of communication and any randomized protocol requires
  $\Omega(\log n)$ of communication for (error bounded by a constant
  strictly less than $\frac{1}{2}$) (see for example \cite{KN06}), both
  our claims follow immediately.
\end{proof}

We now establish our lower bound for {\DCFL}s.
\begin{proof}[Proof of Theorem~\ref{thm:vpl-lowerbound}]
Consider the language $L$ generated by the \CFG\ with rules 
\[ S \rightarrow [S]~\mid~ [S)~ \mid~ (S]~ \mid~ \epsilon.\] 
Note that $L$ is in {\DCFL}; in fact, it is a {\VPL}.

It is easy to verify that two strings $x, y \in \{0,1\}^n$ represent
characteristic vectors of disjoint subsets of $\{1,2,\ldots,n\}$ iff
the string $x' y' \in L$, where $x'$ is obtained from $x$ and $y'$
from $y$ exactly as in the proof of
Theorem~\ref{thm:vpl-lowerbound}. Thus, a $p$-pass space $s$ streaming
algorithm for membership testing in $L$ can be used to derive a
protocol for the set disjointness problem using communication
$(2p-1)s$.  Since the bounded error randomized communication
complexity of the set disjointness problem is $\Omega(n)$ (see
\cite{KN06}), our claim follows immediately.
\end{proof}

\section{Streaming algorithms for checking degree sequence of graphs}
\label{sec:deg-seq}
In this section, we study the complexity of solving the problem \DEGSEQ\ defined
in Section \ref{sec:intro}. We present the proof of the first part of Theorem \ref{thm:deg-seq}.

\begin{proof}[Proof of part $1$ of Theorem \ref{thm:deg-seq}] We come up with
  a uni-variate polynomial from the given degree sequence and the set
  of edges such that the polynomial is identically zero if and only if
  the graph has the given degree sequence.

  We do not store the polynomial explicitly. Instead, we evaluate this
  polynomial at a random point chosen from a large enough field and
  only maintain the evaluation of the polynomial. The Schwartz-Zippel
  lemma \cite{MR96} gives us that with high probability the evaluation
  will be non-zero if the polynomial is non-zero. (If the polynomial
  is identically zero, its evaluation will also be zero.)

  Let the vertex set of the graph be $\{1,\ldots,n\}$. The uni-variate polynomial that we construct is:
$$q(x) = \sum_i d_i x^i- \sum_{i=1}^m x^{u_i} $$ 

The algorithm can be now described as:
\begin{algorithm}
  \begin{algorithmic}
    \STATE Pick $\alpha \in_{R} {\mathbb{F}_p}$($p$ will be fixed
    later). 
    \STATE $Sum \leftarrow 0$  
    \FOR{$i=1$ to $n$} 
    \STATE $Sum\leftarrow Sum + d_i\alpha^i$
    \ENDFOR
    \FOR{$i=1$ to $m$ (where $m$ number of edges)} 
    \STATE $Sum \leftarrow Sum - \alpha^{u_i}$
    \ENDFOR
    \IF{$Sum=0$} 
       \STATE accept 
    \ELSE 
       \STATE reject
    \ENDIF
    \caption{Randomized streaming algorithm for \DEGSEQ}
  \end{algorithmic}
\end{algorithm}

It is easy to note that the algorithm requires only log-space as long
as $p$ is $O(poly(n))$. The input is being read only once from left to
right. For the correctness, note that if the given degree sequence
corresponds to that of the given graph, then $q(x)$ is identically
zero and the value of $Sum$ is also zero for any randomly picked
$\alpha$. We know that $q(x)$ is non-zero when the given degree
sequence does not correspond to that of the given graph. However, the
evaluation may still be zero. Note that degree of $q(x)$ is $n$. If
the field size is chosen to be $n^{1+c} \leq p \leq n^{2+c}$ then due
to Schwartz-Zippel lemma \cite{MR96} the probability that $Sum$ is
zero given that $q(x)$ is non-zero is at most $n/p$ which is at most
$n^{-c}$.
\end{proof}

Now we give a $p$-pass, ${O}((n \log n)/p)$-space deterministic
algorithm for \DEGSEQ\ and hence prove part $2$ of Theorem
\ref{thm:deg-seq}.  The algorithm simply stores the degrees of $n/p$
vertices during a pass and checks whether those vertices have exactly
the degree sequence as stored. If the degree sequence is correct, then
proceed to the next set of $n/p$ vertices. The algorithm needs to
store ${O}((n \log n)/p)$ bits during any pass. The algorithm makes
$p$-passes. 

Finally we show that both the algorithms presented for \DEGSEQ\ is optimal up to 
a $\log n$ factor,
by proving Theorem \ref{thm:deg-seq-lowerbound}.
\begin{proof}[Proof of Theorem \ref{thm:deg-seq-lowerbound}]
  We reduce the two party communication problem of testing equality to
  that of \DEGSEQ. Given strings $x,y\in\{0,1\}^{n}$ we obtain a
  degree sequence $d=(d_{1},d_{2},\cdots,d_{n})$ and a list of edges
  $e_{1}e_{2}\cdots e_{m}$. Take $d_i=x_i$ and for each $i$ such that
  $y_{i}=1$, add an edge $(i,i)$. Clearly $\EQUALITY(x,y)=1$ if
  and only if $d$ is the degree sequence of the graph with edges
  $e_{1}e_{2}\cdots e_{m}$. Again, as in proof of Theorem \ref{thm:dlin-lowerbound}, the theorem follows because of the known communication complexity
  lower bounds for \EQUALITY .
\end{proof}






\bibliographystyle{model1-num-names}
\bibliography{bibliography}







\end{document}